\documentclass[12pt,colorinlistoftodos]{amsart}
\usepackage[utf8]{inputenc}
\usepackage[titletoc]{appendix}
\usepackage[LGR,T1]{fontenc}
\usepackage{bm,bbm}
\usepackage{dsfont}
\usepackage{relsize}
\usepackage{mathtools}
\usepackage{comment}

\usepackage[backend=biber, style=numeric, date=year, url=false, doi=false, isbn=false, eprint=true, firstinits=true, maxcitenames=5, maxbibnames=5, style=alphabetic]{biblatex}
\AtEveryBibitem{%
  \ifentrytype{online}{%
  }{%
    \clearfield{eprint}%
    \clearfield{urldate}%
  }%
}
\usepackage{amsmath}
\usepackage{amssymb}
\usepackage{amsfonts}
\usepackage{amsthm}
\usepackage{tikz}
\usepackage{caption}
\usepackage{graphics}
\usepackage{xcolor,colortbl} 
\usepackage{subfigure}
\usepackage{epigraph}
\usepackage{placeins}

\numberwithin{equation}{subsection}

\usepackage{chngcntr}
\counterwithin{table}{subsection}
\counterwithin{figure}{subsection}

\usepackage{enumerate}
\usepackage{verbatim}
\usetikzlibrary{trees}
\usetikzlibrary{decorations.markings}
\usetikzlibrary{arrows}
\usetikzlibrary{cd}
\usepackage{mathrsfs}
\usepackage{xparse}

\usepackage{relsize}
\usepackage{url}
\usepackage{hyperref}
\hypersetup{colorlinks=true, citecolor=purple, linktocpage, linkcolor=blue, urlcolor=cyan} 

\theoremstyle{plain}
\newtheorem{thm}{Theorem}[subsection]
\newtheorem{lem}[thm]{Lemma}
\newtheorem{prop}[thm]{Proposition}

\theoremstyle{definition}
\newtheorem{defn}[thm]{Definition}

\newtheorem*{thm*}{Theorem}
\newtheorem*{lem*}{Lemma}
\newtheorem*{prop*}{Proposition}
\newtheorem*{cor*}{Corollary}
\newtheorem*{exe*}{Exercise}
\newtheorem*{defn*}{Definition}
\newtheorem{rem}[thm]{Remark}

\theoremstyle{remark}

\newcommand{\R}{\mathbb{R}}
\newcommand{\Z}{\mathbb{Z}}

\newcommand{\calY}{\mathcal{Y}}

\newcommand{\ii}{{\mathrm{i}}}
\newcommand{\dd}{{\mathrm{d}}}

\DeclareMathOperator{\gh}{gh} 
 
\newcommand{\de}{\partial}

\newcommand{\calB}{\mathcal{B}}
\newcommand{\calH}{\mathcal{H}}
\newcommand{\calS}{\mathcal{S}}

\newcommand{\calG}{\mathcal{G}}

\newcommand{\calL}{\mathcal{L}}
\newcommand{\calM}{\mathcal{M}}
\newcommand{\calQ}{\mathcal{Q}}

\newcommand{\calT}{\mathcal{T}}

\newcommand{\calP}{\mathcal{P}}
\newcommand{\calF}{\mathcal{F}}

\def\gpd{\,\lower1pt\hbox{$\longrightarrow$}\hskip-.24in\raise2pt
               \hbox{$\longrightarrow$}\,}

\let\Tilde=\widetilde

\let\Hat=\widehat

\newcommand{\I}{\mathrm{i}}

\newcommand{\calV}{\mathcal{V}}

\makeatletter
\newcommand{\upint}{\DOTSI\upintop\ilimits@}
\newcommand{\upoint}{\DOTSI\upointop\ilimits@}
\makeatother

\usepackage[colorinlistoftodos]{todonotes}

\makeatletter

\pgfdeclareshape{genus}{
 \anchor{center}{\pgfpointorigin}
\backgroundpath{
    \begingroup
    \tikz@mode
    \iftikz@mode@fill

         \pgfpathmoveto{\pgfqpoint{-0.78cm}{-.17cm}}
     \pgfpathcurveto %
            {\pgfpoint{-0.35cm}{-.44cm}}
        {\pgfpoint{0.35cm}{-.44cm}}
        {\pgfpoint{.78cm}{-0.17cm}} 
     \pgfpathmoveto{\pgfqpoint{-0.78cm}{-0.17cm}}
     \pgfpathcurveto %
            {\pgfpoint{-0.25cm}{.25cm}}
        {\pgfpoint{.25cm}{.25cm}}
        {\pgfpoint{0.78cm}{-0.17cm}}
        \pgfusepath{fill}
        \fi

        \iftikz@mode@draw
     \pgfpathmoveto{\pgfqpoint{-1cm}{0cm}}
     \pgfpathcurveto %
            {\pgfpoint{-0.5cm}{-.5cm}}
        {\pgfpoint{0.5cm}{-.5cm}}
        {\pgfpoint{1cm}{0cm}}

     \pgfpathmoveto{\pgfqpoint{-0.75cm}{-0.15cm}}
     \pgfpathcurveto %
            {\pgfpoint{-0.25cm}{.25cm}}
        {\pgfpoint{.25cm}{.25cm}}
        {\pgfpoint{0.75cm}{-0.15cm}}
              \pgfusepath{stroke}
        \fi
        \endgroup
    }
    }

    \makeatother

\tikzset{residual/.style={draw, shape=circle, black,inner sep=1pt}}

\title[Equivariant BV-BFV Formalism]{Equivariant BV-BFV Formalism}
\author[A. S. Cattaneo]{Alberto S. Cattaneo}
\author[N. Moshayedi]{Nima Moshayedi}
\address{Institut f\"ur Mathematik\\ Universit\"at Z\"urich\\ 
Winterthurerstrasse 190
CH-8057 Z\"urich}
\email[A.~S.~Cattaneo]{cattaneo@math.uzh.ch}
\address{Institut f\"ur Mathematik\\ Universit\"at Z\"urich\\ 
Winterthurerstrasse 190
CH-8057 Z\"urich}
\email[N.~Moshayedi]{nima.moshayedi@math.uzh.ch}
\date{\today}
\thanks{ASC acknowledges partial support of the SNF Grant No.\ 200021\_227719 and of the Simons Collaboration on Global Categorical Symmetries. This research was (partly) supported by the NCCR SwissMAP, funded by the Swiss National Science Foundation. This article is based upon work from COST Action 21109 CaLISTA, supported by COST (European Cooperation in Science and Technology) (www.cost.eu), MSCA-2021-SE-01-101086123 CaLIGOLA, and MSCA-DN CaLiForNIA - 101119552.}

\begin{document}
\makeatletter
\providecommand\@dotsep{5}
\def\listtodoname{List of Todos}
\def\listoftodos{\@starttoc{tdo}\listtodoname}
\makeatother

\begin{abstract}
   The recently introduced equivariant BV formalism is extended to the case of manifolds with boundary under appropriate conditions. AKSZ theories are presented as a practical example.  
\end{abstract}

\maketitle

\tableofcontents

\section{Introduction}
The BV formalism is a general method to describe field theory with symmetries. In the presence of a boundary of space--time, it gives rise---under some regularity assumptions---to a nice coupling with the boundary BFV formalism (the cohomological description of the reduced phase space). It was studied in \cite{CMR1} and called the BV-BFV formalism (see also \cite{CattMosh1} for an introduction). This is also the starting point for perturbative quantization in the presence of a boundary \cite{CMR3,CMR2}.

In \cite{CattZabz2019} the BV formalism was extended to the equivariant case arising when a vector field 
(more generally, an involutive family of vector fields) on space--time is given. 

In this paper, we study the extension of the equivariant BV formalism in the presence of a boundary and obtain the conditions under which the construction leading to the equivariant BV-BFV formalism works.

To summarize the result, just recall that, in its most general formulation, what is studied in \cite{CattZabz2019} is a violation of the classical master equation of the BV formalism. Equivalently, one drops the condition 
that the BV charge $\mathcal{Q}$ (the Hamiltonian vector field of the BV action) is cohomological (i.e., $[\calQ,\calQ]=0$). We show that the BV-BFV construction works through if $[\calQ,\calQ]$ is a symplectic vector field (or, equivalently, a Hamiltonian vector field) also in the presence of a boundary (in the absence of a boundary this is automatically true).

The particularly important case of the equivariant extension \cite{CattZabz2019}
of AKSZ theories \cite{AKSZ} then turns out to be compatible with the boundary if and only if all the given vector fields in the involutive family are tangent to the boundary.

The original results of this paper are in Section~\ref{s-eqBVBFV}. In the previous sections we recall background material: Section~\ref{s-BVBFV} covers the BV formalism, the BV-BFV formalism, the equivariant BV formalism, and their quantization, whereas Section~\ref{s-AKSZ} gives more details on the AKSZ construction.

\textbf{Acknowledgements} We thank Francesco Bonechi for 
providing valuable comments and remarks. We also thank an anonymous referee for helpful comments. 

\section{The BV-BFV Formalism}\label{s-BVBFV}
\subsection{Classical BV Formalism} We start by recalling the classical setting of the BV formalism.

\begin{defn}[Lagrangian field theory]
A $d$-dimensional \emph{Lagrangian field theory} assigns to a $d$-manifold $M$ a pair $(F_M,S_M)$, where $F_M$ is some \emph{space of fields}\footnote{An important class of examples of $F_M$ is given by the space of sections of a vector bundle over $M$, e.g., vector fields or differential forms on $M$. However, it can be more complicated such as the space of connections or metrics on $M$.} where elements are sections of some fibre bundle over $M$ and $S_M\in C^\infty(F_M)$ is a function on $F_M$, called the \emph{action (functional)}. We say that a Lagrangian field theory is \emph{local} if we can express $S_M$ as 
\[
S_M(\phi)=\int_M\mathscr{L}(j^k\phi),\quad \phi\in F_M,
\]
where $\mathscr{L}$ denotes a \emph{Lagrangian density} on $M$ and $j^k\phi$ denotes the $k$-th jet prolongation of the field $\phi\in F_M$ as a section of the jet bundle of $F_M$.
\end{defn}

\begin{defn}[BV manifold]
A \emph{BV manifold} consists of a triple $(\calF,\calS,\omega)$, where $\calF$ is a $\Z$-graded\footnote{The $\Z$-grading is called the \emph{ghost number} in the physics literature according to the construction by Faddeev and Popov \cite{FP}.} supermanifold,\footnote{The supermanifold structure induces an additional $\Z_2$-grading which is referred to as \emph{parity}, often called \emph{odd} and \emph{even}.} $\calS$ an even function on $\calF$ of degree 0 and $\omega$ an odd symplectic form on $\calF$ of degree $-1$, such that the \emph{classical master equation (CME)} holds: 
\begin{equation}
\label{eq:CME}
(\calS,\calS)=0.
\end{equation}
Here, $(\enspace,\enspace)$ denotes the odd Poisson bracket of degree $+1$ induced by $\omega$.
\end{defn}

\begin{rem}
We call $\calF$ the \emph{BV space of fields}, $\calS$ the \emph{BV action (functional)} and $\omega$ the \emph{BV symplectic form}. The odd Poisson bracket $(\enspace,\enspace)$ is usually called the \emph{BV bracket} (or \emph{anti-bracket}).\footnote{The BV bracket is usually denoted by round brackets in order to emphasize the fact that it is an \emph{odd} Poisson bracket. This notation is originally due to Batalin and Vilkovisky \cite{BV2}.} We will denote by $\mathrm{gh}(\sigma)\in\Z$ the $\Z$-grading and by $\vert \sigma\vert\in \Z_2$ the parity. If $\sigma$ is also a differential form on $M$, we will denote by $\deg\sigma$ its form degree. 
\end{rem}

We also consider the Hamiltonian vector field $\calQ$ of $\calS$, i.e., the unique vector field of degree $+1$ satisfying the equation 
\[
\iota_Q\omega=\delta\calS,
\]
where $\delta$ here denotes the de Rham differential on $\calF$. It is easy to see that the vector field $\calQ$ is actually \emph{cohomological}, i.e., we have $\calQ^2=0$ and moreover, by definition, it is \emph{symplectic}, i.e., we have $L_\calQ\omega=0$. Here $L$ denotes the Lie derivative. Note that, by definition, we have that $\calQ=(\calS,\enspace)$ and hence by the CME \eqref{eq:CME} we get $L_\calQ\calS=\calQ\calS=0$. We will call $\calQ$ the \emph{BV charge}.\footnote{The cohomological vector field $\calQ$ is the same as the one in the BRST formalism \cite{BRS1,BRS2,Tyutin1976} measuring the symmetries of the theory, thus often also called the \emph{BRST charge}.}

\begin{defn}[BV theory]
A $d$-dimensional \emph{BV theory} is the assignment $M\mapsto (\calF_M,\calS_M,\omega_M)$ of a closed, connected $d$-manifold $M$ to a BV manifold. 
\end{defn}

\subsubsection{2D $BF$ theory}
Let $M$ be a connected closed 2-manifold and let $G$ be a Lie group with Lie algebra $\mathfrak{g}$. Let $P\to M$ be a trivial principal $G$-bundle over $M$. Consider the space of connection 1-forms $\Omega^1(M,\mathrm{ad} P)$.\footnote{We work around the trivial connection.} We define the space of fields for 2-dimensional $BF$ theory to be 
\[
F_M:=\Omega^1(M,\mathrm{ad} P)\oplus \Omega^{0}(M,\mathrm{ad}^*P),
\]
where $\mathrm{ad}^*P$ denotes the coadjoint bundle of $P$. The $BF$ action is given by 
\begin{equation}
\label{eq:BF_action}
    S_M:=\int_M\langle B, F_A\rangle=\int_M\left(\langle B,\dd A\rangle +\frac{1}{2}\langle B,[A,A]\rangle\right),
\end{equation}
where $(A,B)\in F_M$ and $F_A:=\dd A+\frac{1}{2}[A,A]$ denotes the curvature 2-form of the connection 1-form $A$. We have denoted by $\langle\enspace,\enspace\rangle$ the pairing between the forms of adjoint and coadjoint type as an extension of the pairing between $\mathfrak{g}$ and $\mathfrak{g}^*$. It is easy to see that the critical points of $S_M$ are given by pairs $(A,B)\in F_M$ where $A$ is a flat connection and $\dd_AB=0$. Here $\dd_A$ denotes the covariant derivative with respect to the connection $A$. 
Denote by $\calG$ the group of gauge transformations of the space of connection 1-forms on $P$ and let $A^g$ denote the gauge transformed connection for a gauge transformation $g\in\calG$. We can let $\calG$ also act on $\Omega^0(M,\mathrm{ad}^*P)$ by the coadjoint action. 
It is then not hard to see that the $BF$ action \eqref{eq:BF_action} is invariant with respect to the gauge transformation
\begin{align}
    A&\mapsto A^g,\\
    B&\mapsto B^{g}:=\mathrm{Ad}^*_{g^{-1}}B, 
\end{align}
i.e., we have 
\begin{multline*}
S_M(A^g,B^{g})=\int_M\left(\langle B^{g},\dd A^g\rangle +\frac{1}{2}\langle B^{g},[A^g,A^g]\rangle\right)\\
=\int_M\left(\left\langle \mathrm{Ad}^*_{g^{-1}}B, 
\dd A^g\right\rangle +\frac{1}{2}\left\langle \mathrm{Ad}^*_{g^{-1}}B, 
[A^g,A^g]\right\rangle\right)\\
=\int_M\left(\langle B,\dd A\rangle +\frac{1}{2}\langle B,[A,A]\rangle\right)=S(A,B).
\end{multline*}

If $\mathfrak{g}=\R$, we speak of \emph{abelian} $BF$ theory. In this case we have 
\[
F_M=\Omega^1(M)\oplus \Omega^0(M),
\]
and
\[
S_M=\int_M B\land \dd A.
\]
The critical points here are pairs $(A,B)\in F_M$ of the form $\dd A=0$ 
and $\dd B=0$. 

The BV formulation of 2-dimensional $BF$ theory is not hard to construct. The BV space of fields associated to the 2-manifold $M$ is given by 
\[
\calF_M:=\Omega^\bullet(M,\mathrm{ad}P)[1]\oplus \Omega^\bullet(M,\mathrm{ad}^*P),
\]
where $\Omega^\bullet=\bigoplus_{j=0}^2\Omega^j$.
The \emph{superfields} are pairs $(\mathbf{A},\mathbf{B})\in\calF_M$ of the form 
\begin{align}
    \mathbf{A}&:=c+A+B^+,\\
    \mathbf{B}&:=B+A^++c^+,
\end{align}
where $\mathrm{gh}(c)=1$, $\mathrm{gh}(A)=\mathrm{gh}(B)=0$, $\mathrm{gh}(A^+)=\mathrm{gh}(B^+)=-1$, $\mathrm{gh}(c^+)=-2$. Moreover, $\deg(c)=\deg(B)=0$, $\deg(A)=\deg(A^+)=1$, $\deg(B^+)=\deg(c^+)=2$. For a classical field $\phi$, we have denoted by $\phi^+$ its \emph{anti-field} with the property $\mathrm{gh}(\phi)+\mathrm{gh}(\phi^+)=-1$ and $\deg(\phi)+\deg(\phi^+)=2$. Note that we have denoted by $c$ the \emph{ghost field}.
We can define the curvature of the superconnection $\mathbf{A}$ by $\mathbf{F}_\mathbf{A}$ defined by 
\[
\mathbf{F}_\mathbf{A}:=F_{A_0}+\dd_{A_0}\mathbf{a}+\frac{1}{2}[\mathbf{a},\mathbf{a}],
\]
where here $[\enspace,\enspace]$ denotes the induced bracket on the super Lie algebra, $A_0$ is some reference connection 1-form and $\mathbf{a}:=\mathbf{A}-A_0\in \Omega^\bullet(M,\mathrm{ad}P)[1]$. The BV action is then given by 
\begin{equation}
\label{eq:BV_BF_action}
\calS_M(\mathbf{A},\mathbf{B}):=\int_M\langle \mathbf{B},\mathbf{F}_\mathbf{A}\rangle,
\end{equation}
where $\langle\enspace,\enspace\rangle$ here denotes the pairing between the forms of adjoint and coadjoint type as an extension of the pairing between $\mathfrak{g}$ and $\mathfrak{g}^*$ with shifted degree, i.e., with additional sign coming from the $\Z$-grading. It is then not hard to see that $\calS_M$ satisfies the CME $(\calS_M,\calS_M)=0$. Note also that $\calQ_M\mathbf{A}=(\calS_M,\mathbf{A})=\mathbf{F}_\mathbf{A}$ and $\calQ_M\mathbf{B}=(\calS_M,\mathbf{B})=\dd_\mathbf{A}\mathbf{B}$. In particular, the cohomological vector field of degree $+1$ is given by
\begin{multline}
\calQ_M=(\calS_M,\enspace)=\int_M \left(\dd\mathbf{A}\frac{\delta}{\delta \mathbf{A}}+\dd\mathbf{B}\frac{\delta}{\delta\mathbf{B}}\right)\\
=\int_M\Bigg(\dd c\frac{\delta}{\delta c}+\dd A\frac{\delta}{\delta A}+\dd B^+\frac{\delta}{\delta B^+}+\dd B\frac{\delta}{\delta B}+\dd A^+\frac{\delta}{\delta A^+}+\dd c^+\frac{\delta}{\delta c^+}\Bigg).
\end{multline}
The BV symplectic form of degree $-1$ is given by 
\[
\omega_M=\int_M\delta\mathbf{A}\land \delta\mathbf{B}
=\int_M\left(\delta c\land \delta c^++\delta A\land \delta A^++\delta B+\delta B^+\right).
\]

\begin{rem}
The general structure of $BF$ theory can be easily generalized to arbitrary dimension and is an example of an AKSZ theory (see Section~\ref{s-AKSZ}).
\end{rem}

\subsection{Quantum BV Formalism}
In the quantum setting, we consider the operator $\Delta$, the canonical \emph{BV operator} (or \emph{BV Laplacian}) introduced by Khudaverdian in \cite{Khudaverdian2004}, acting on half-densities and
formally extended to the infinite-dimensional setting. There are several ways of characterizing this operator. One characterization is to choose a symplectomorphism $\calF\to T^*[-1]N$, for some manifold $N$, invoking the global Darboux theorem for odd symplectic manifolds \cite{S}. Then we can observe that half-densities are canonically identified with differential forms on $N$. Thus, we can consider $\Delta$ as the de Rham differential on $N$. Another characterization, on a finite-dimensional odd-symplectic supermanifold $(\calM,\omega)$ with local coordinates $(x_i,\theta^i)$, where $x_i$ denote the even coordinates and $\theta^i$ denote the corresponding odd momenta (i.e., $\omega=\sum_{i}\dd \theta_i\land \dd x^i$), is that it is given by 
\[
\Delta=\sum_{i}\frac{\de^2}{\de x_i\de\theta^i}.
\]
Its fundamental algebraic property is $\Delta^2=0$.

We often need the BV operator to act on functions. For this we pick a nowhere vanishing half-density $\mu$ satisfying $\Delta\mu=0$. Such a reference half-density exists (e.g., by the symplectomorphism from $\calF\to T^*[-1]N$). 
Using this we can define a BV operator $\Delta_\mu$ acting on functions via
\[
\Delta_\mu f \coloneqq (\Delta f\mu)/\mu.
\]
This operator satisfies 
$\Delta^2=0$ and 
\[
\Delta_\mu(fg)=\Delta_\mu fg\pm f\Delta_\mu g \pm(f,g).
\]
It can also be computed by
\begin{equation}
\label{eq:Laplacian}
\Delta_\mu f=\frac{1}{2}\mathrm{div_{\mu^2}}(f,\enspace),\quad f\in C^\infty(\calF),
\end{equation}
where $\mathrm{div}_{\mu^2}$ denotes the divergence operator on vector fields determined by the density $\mu^2$.
When the reference half-density is understood, we will simply write $\Delta$ also for the BV operator acting on functions.

See also \cite{Cattaneo2023} for a detailed exposition on the BV Laplacian and \cite{Severa2006} for another mathematical interpretation.

    An important observation by Schwarz \cite{S} (in the finite-dimensional case) is that, assuming orientation, a half-density on $\calF$ determines a density on every Lagrangian submanifold of $\calF$. 

\begin{thm}[Batalin--Vilkovisky--Schwarz\cite{BV2,S}]\label{thm:BV}
Consider two half-densities $f,g$ on an odd-symplectic supermanifold $(\calM,\omega)$. Then\footnote{We assume below that the given half-densities 
have fast decay at infinity.} 
\begin{enumerate}
    \item if $f=\Delta g$ (\emph{BV exact}), we get that
    \[
    \int_{\calL\subset\calM} f=\int_{\calL\subset \calM} \Delta g=0, 
    \]
    for any Lagrangian submanifold $\calL\subset \calM$.
    \item if $\Delta f=0$ (\emph{BV closed}), we get that
    \[
    \frac{\dd}{\dd t}\int_{\calL_t\subset \calM}f=0,
    \]
    for any differentiable family $(\calL_t)$ of Lagrangian submanifolds of $\calM$. In particular, when two Lagrangian submanifolds $\calL_1\subset \calM$ and $\calL_2\subset \calM$ can be 
    deformed into each other, we have 
    \[
    \int_{\calL_1\subset \calM}f=\int_{\calL_2\subset\calM}f.
    \]
\end{enumerate}
\end{thm}

For the setting of quantum gauge field theories, we are mainly interested in the case where our odd-symplectic supermanifold is given by a pair $(\calF,\omega)$ and the half-density $f$ is of the form 
\[
\exp\left(\frac{\I}{\hbar} S\right)\rho\in \mathrm{Dens}^{\frac{1}{2}}(\calF),
\]
where $\rho$ denotes a reference, $\Delta$-closed, nowhere vanishing half-density on $\calF$. Then, when considering an integral of the form $\int_\calL \exp\left(\frac{\I}{\hbar}S\right)\rho$, the choice of Lagrangian submanifold corresponds to choosing a gauge fixing\footnote{When choosing a gauge-fixing fermion $\Psi$, the case where we can reduce to the cohomological setting of the BRST formalism, is given by taking the Lagrangian submanifold to be the graph of the differential of the gauge-fixing fermion, i.e., we take $\calL=\mathrm{graph}(\dd\Psi)$.}.
Note that the second point of Theorem \ref{thm:BV} is then a condition for gauge  independence of the classical theory described by the action $S$. In the case of interest, we want thus
\begin{equation}
\label{eq:QME}
\Delta \exp\frac{\I}{\hbar}S=0\Longleftrightarrow \frac{1}{2}(S,S)-\I\hbar\Delta S=0.
\end{equation}
Equation \eqref{eq:QME} is called the \emph{quantum master equation (QME)}. Note that in the semi-classical limit $\hbar\to 0$ we get the classical master equation (CME) $(\calS,\calS)=0$. When $S$ depends on $\hbar$ as a formal power series $S=S_0+\hbar S_1+\hbar^2S_2+\dotsm$, we can try to solve the QME order by order in $\hbar$. 

When we consider the BV action $\calS$, which gives a solution to the CME \eqref{eq:CME}, we often assume that $\Delta\calS=0$, in which case $\calS$ then also solves the QME. This is true for many different theories of interest but is not immediate. In general, we need to find a suitable regularization in order to make the QME hold. Note also that by Equation \eqref{eq:Laplacian}, we have 
\[
\Delta \calS=\frac{1}{2}\mathrm{div}\,\calQ,
\]
so the condition is that the reference half-density $\rho$ is $Q$-invariant.

\subsection{Classical BV-BFV Formalism}
In the BV formalism the source manifold $M$ was always assumed to have empty boundary. In order to overcome this and still make sense of a gauge formalism for manifolds with boundary, one couples the Lagrangian approach of the BV theory in the bulk to the Hamiltonian approach of the \emph{BFV theory}\footnote{The letters BFV stand for Batalin--Fradkin--Vilkovisky due to their work \cite{BF1,BF2,FV1,FV2} where they developed the Hamiltonian setting.} on the boundary such that everything is coherent at the end. The coupling of bulk and boundary was considered in the classical setting first in \cite{CMR1}.

\begin{defn}[BFV manifold]
A \emph{BFV manifold} is a quadruple 
\[
(\calF^\de,\omega^\de,\calS^\de,\calQ^\de)
\]
where $\calF^\de$ is a $\Z$-graded supermanifold, $\calS^\de$ is an odd function on $\calF^\de$ of degree $+1$, $\omega$ is an even symplectic form on $\calF^\de$ of degree 0, and $\calQ^\de$ is a cohomological vector field on $\calF^\de$ of degree $+1$ such that $\iota_{\calQ^\de}\omega^\de=\delta \calS^\de$. 
This implies that $\{\calS^\de,\calS^\de\}_{\omega^\de}=0$, where $\{\enspace,\enspace\}_{\omega^\de}$ denotes the even Poisson bracket\footnote{We will often just write $\{\enspace,\enspace\}$ without emphasizing on $\omega^\de$ whenever it is clear. Moreover, we will also call this the \emph{boundary} Poisson bracket.} of degree 0 induced by the symplectic form $\omega^\de$.
We say that a BFV manifold is \emph{exact}, if the symplectic form is exact, i.e., $\omega^\de=\delta\alpha^\de$.
\end{defn}

\begin{rem}
We call $\calF^\de$ the \emph{BFV space of fields}, $\calS^\de$ the \emph{BFV action (functional)}, $\omega^\de$ the \emph{BFV symplectic form} and $\calQ^\de$ the \emph{BFV charge}. If we consider the exact case, we will call $\alpha^\de$ the \emph{BFV 1-form}.
\end{rem}

\begin{defn}[BV-BFV manifold over exact BFV manifold]
A \emph{BV-BFV manifold over an exact BFV manifold} $(\calF^\de,\omega^\de=\delta\alpha^\de,\calS^\de,\calQ^\de)$ is a quintuple 
\[
(\calF,\omega,\calS,\calQ,\pi),
\]
where the objects are defined as follows: $\calF$ is a $\mathbb{Z}$-graded supermanifold, $\omega$ is an odd symplectic form on $\calF$ of degree $-1$, $\calS$ is an even function on $\calF$ of degree $0$, $\calQ$ is a cohomological vector field of degree $+1$, and $\pi\colon \calF\to \calF^\de$ is a surjective submersion, such that 
\begin{enumerate}
    \item $\iota_\calQ\omega=\delta\calS+\pi^*\alpha^\de$,
    \item $\delta\pi \calQ=\calQ^\de$,
\end{enumerate}
where $\delta\pi$ denotes the pushforward of the surjective submersion $\pi$.
\end{defn}

Note that, in the closed setting, we have required that the CME $L_\calQ\calS=0$ holds, whereas in the setting with boundary the Lie derivative $L_\calQ\calS$ is not zero anymore, but given entirely in terms of boundary data.
Indeed, one can easily prove \cite{CMR3} that 
\begin{equation}
\label{eq:mCME}
L_\calQ\calS=\pi^*\left(2\calS^\de-\iota_{\calQ^\de}\alpha^\de\right)
\end{equation}
and
\begin{equation}\label{e:mCME2}
\frac12 \iota_\calQ \iota_\calQ  \omega = \pi^* \calS^\de.
\end{equation}
Both equations express the violation of the CME by boundary terms.

We recall their proofs because in the equivariant case we will generalize them.
We know that
\[
L_\calQ=\iota_\calQ\delta-\delta \iota_\calQ
\]
and that the same holds for $\calQ^\de$. We will now apply $L_\calQ$ to the equation $\iota_\calQ\omega=\delta\calS+\pi^*\alpha^\de$. By the definition of $L_\calQ$ and since $\calQ^2=0$, we get $[L_\calQ,\iota_\calQ]=L_\calQ\iota_\calQ-\iota_\calQ L_\calQ=0$ and we know that 
\[
L_\calQ\omega=\pi^*\delta\alpha^\de=\pi^*\omega^\de.
\]
Thus, we get
\begin{equation}\label{eq:mCME_1}
L_\calQ\iota_\calQ\omega=\iota_\calQ L_\calQ\omega=\pi^*\iota_{\calQ^\de}\delta\alpha^\de.
\end{equation}
If we then apply the Lie derivative $L_\calQ$ to the CME, we get
\begin{multline*}
L_\calQ\iota_\calQ\omega=L_\calQ\delta\calS+L_\calQ\pi^*\alpha^\de=\delta\iota_\calQ\delta\calS+\pi^*(L_{\calQ^\de}\alpha^\de)\\
=\delta L_\calQ\calS+\pi^*\Big(\iota_{\calQ^\de}\delta\alpha^\de-\delta\iota_{\calQ^\de}\alpha^\de\Big).
\end{multline*}
In the last line, we used Cartan's magic formula for $L_{\calQ^\de}$. Now, using \eqref{eq:mCME_1} and  $\delta\calS^\de=\iota_{\calQ^\de}\omega^\de$, we get
\[
\pi^*(\delta\calS^\de)=\delta L_{\calQ}\calS+\pi^*(\delta \calS^\de)-\pi^*(\delta\iota_{\calQ^\de}\alpha^\de),
\]
which is equivalent to 
\[
\delta\left(L_\calQ\calS-\pi^{*}(2\calS^\de-\iota_{\calQ^\de}\alpha^\de)\right)=0.
\]
This means that function in the argument of $\delta$ is constant. However, since it is a function of degree $+1$ is zero, it must vanish. Therefore, we get \eqref{eq:mCME}
proving that the CME is spoiled only up to boundary data. To get \eqref{e:mCME2}, simply apply $\iota_\calQ$ to $\iota_\calQ\omega=\delta\calS+\pi^*\alpha^\de$ and use \eqref{eq:mCME}.

We call Equation \eqref{eq:mCME} the \emph{modified} CME \cite{CMR1}. Condition (2) tells us that the BV charge is projectable onto the BFV charge $\calQ^\de$. For examples and more insights on the classical BV-BFV formalism we refer to \cite{CMR1,CattMosh1}.

\subsection{Quantum BV-BFV Formalism}
As it was developed in \cite{CMR2} (see also \cite{CattMosh1}), we can express a quantum BV-BFV formalism by the following collection of data for a source manifold $M$ with boundary $\de M$:

\begin{enumerate}[$(i)$]
    \item The \emph{state space} $\calH^\calP_{\de M}$, a graded vector space associated to the boundary $\de M$ with a choice of a \emph{polarization}\footnote{A \emph{polarization} is the choice of an involutive Lagrangian subbundle of the tangent bundle of a manifold. In our case, $\calP$ is an involutive Lagrangian subbundle of $T\calF^\de_{\de M}$. Actually, a polarization is needed in order to construct the state space by techniques of \emph{geometric quantization} \cite{Wood97,BatesWeinstein2012}.} $\calP$ on $\calF^\de_{\de M}$. 
    \item The \emph{quantum BFV operator} $\Omega^\calP_{\de M}$, a coboundary operator on the state space $\calH^\calP_{\de M}$ which is determined as a quantization of the BFV action $\calS^\de_{\de M}$.
    \item The \emph{space of residual fields} $\calV_M$, a finite-dimensional graded manifold endowed with a symplectic form of degree $-1$ associated to $M$ and a polarization $\calP$ on $\calF^\de_{\de M}$. Moreover, we can define the graded vector space\footnote{We denote by $\widehat{\otimes}$ the completed tensor product.} 
    \[
    \calH^\calP_M:=\calH^\calP_{\de M}\widehat{\otimes}\mathrm{Dens}^\frac{1}{2}(\calV_M),
    \]
    where $\mathrm{Dens}^\frac{1}{2}(\calV_M)$ denotes the space of half-densities on $\calV_M$. This graded vector space is endowed with two commuting coboundary operators 
    \begin{align*}
        \Omega^\calP_M&:=\Omega^\calP_{\de M}\otimes \mathrm{id},\\
        \Delta^\calP_M&:=\mathrm{id}\otimes \Delta_\mathrm{res},
    \end{align*}
    where $\Delta_\mathrm{res}$ denotes the canonical BV Laplacian on half-densities on residual fields.
    \item A \emph{state} $\psi_M\in\calH^\calP_M$ that satisfies the \emph{modified quantum master equation (mQME)}
    \[
    \left(\hbar^2\Delta^\calP_M+\Omega^\calP_M\right)\psi_M=0.
    \]
    This equation is the quantum version of the modified CME \eqref{eq:mCME}.
\end{enumerate}

\section{AKSZ Theories}\label{s-AKSZ}
\subsection{Short Overview} An important class of BV theories were developed by Alexandrov, Kontsevich, Schwarz and Zaboronsky in \cite{AKSZ} and known today as AKSZ theories. They formulated a method  to obtain a solution of the CME by considering a special form for the space of fields, namely, a mapping space (actually, a space of internal homs) between graded manifolds. Let us first describe the ingredients needed to formulate these type of theories. Let $M$ be a closed, connected $d$-manifold and consider a differential graded symplectic manifold $(\calM,\omega)$ where the symplectic form is exact, i.e., $\omega=\dd \alpha$, and of degree $d-1$. Moreover, let $\Theta$ be a function on $\calM$ of degree $d$ such that $\{\Theta,\Theta\}_\omega=0$, where $\{\enspace,\enspace\}_\omega$ denotes the Poisson bracket of degree $1-d$ induced by $\omega$. Moreover, consider the Hamiltonian vector field $\calQ_{\mathcal{M}}$ for the Hamiltonian function $\Theta$ which is cohomological by definition, i.e., $(\calQ_\mathcal{M})^2=0$. Define the space of fields as 
\[
\calF^\mathrm{AKSZ}_M:=\mathrm{Map}(T[1]M,\calM).
\]
Using the symplectic structure on $\calM$, we can construct a symplectic structure on $\calF^\mathrm{AKSZ}_M$ by \emph{transgression}. Namely, we have a diagram 
\[
\begin{tikzcd}
{\mathrm{Map}(T[1]M,\mathcal{M})\times T[1]M} \arrow[d, "\pi"'] \arrow[r, "\mathsf{ev}"] & \mathcal{M} \arrow[ld, dashed] \\
{\mathrm{Map}(T[1]M,\mathcal{M})}                                                               &                               
\end{tikzcd}
\]
where $\mathsf{ev}$ denotes the evaluation map and $\pi$ the projection onto the first factor. Thus, on the level of forms, we can define a \emph{transgression map} $\mathbb{T}\colon \Omega^\bullet(\calM)\to \Omega^\bullet(\mathrm{Map}(T[1]M,\calM))$ by 
\[
\mathbb{T}(\eta):=\pi_*\mathsf{ev}^*\eta=\int_{T[1]M}\mathsf{ev}^*\eta,\quad \eta\in \Omega^\bullet(\calM).
\]
Hence, we define a symplectic form on $\calF^\mathrm{AKSZ}_M$ by 
\[
\omega_M:=\mathbb{T}(\omega)=\int_{T[1]M}\mathsf{ev}^*\omega,
\]
where $\omega$ is the symplectic form on $\calM$. Note that since $\omega$ was of degree $d-1$, we get that $\omega_M$ is of degree $(d-1)-d=-1$ which is the correct degree for a BV symplectic form. 
We can define a cohomological vector field on the mapping space as the sum of the lift of the de Rham differential on $M$ and the lift of the cohomological vector field on $\calM$:
\[
\calQ_M:=\widehat{\dd_M}+\widehat{\calQ_\mathcal{M}},
\]
where the hats denote the lift to the mapping space.
The BV action is then constructed by using the primitive $\alpha$ of $\omega$ and the Hamiltonian function $\Theta$:
\[
\calS^\mathrm{AKSZ}_M=\iota_{\widehat{\dd_M}}\mathbb{T}(\alpha)+\mathbb{T}(\Theta).
\]
It is then not difficult to see that $\calS^\mathrm{AKSZ}_M$ is indeed of degree $0$ and satisfies the CME.

\begin{rem}
A lot of relevant theories are of AKSZ type, such as e.g., Chern--Simons theory \cite{Chern1974,Witten1989,AS,AS2}, the Poisson sigma model \cite{I,SS1,CF4}, or Witten's $A$- and $B$-twisted sigma models \cite{Witten1988a,AKSZ}. 
\end{rem}

\section{Equivariant BV-BFV Formalism}\label{s-eqBVBFV}

\subsection{Equivariant BV Formalism}
In this section we recollect the results of
\cite{CattZabz2019} on the equivariant BV formalism.
The easiest way to prepare the background material is to start from the quantum version. We begin with the very general observation of what happens when the QME is violated. Namely, we introduce a functional $\calT$, with $\gh(\calT)=1$, that measures the failure of the QME:
\[
\calT\coloneqq\left(\frac\hbar\ii\right)^2\exp\left(-\frac{\I}{\hbar}\calS\right)\Delta\exp\frac{\I}{\hbar}\calS=\frac12(\calS,\calS)-\ii\hbar\Delta \calS.
\]
It turns out that the properties of the BV integral are still preserved if we only allow $\calT$-Lagrangian submanifolds, i.e., Lagrangian submanifolds on which $\calT$ vanishes.

If we in addition assume that $\Delta \calS=0$, as is often the case, or just ignore higher orders in $\hbar$, then we simply get
\begin{equation}\label{e:SS2T}
(\calS,\calS) = 2\calT,
\end{equation}
which expresses $\calT$ as the violation of the CME. This is the equation we are going to study in the presence of boundary.

\subsection{Preliminary Computations}
We want to start with some important definitions and relations in order to understand the the core of the computation for the case with boundary.

\begin{defn}[Weak BV manifold]
    A \emph{weak BV manifold} of degree $k$ is a quadruple $(\calM, \calQ, \omega, \calS)$, where $\calM$ is a graded manifold, $\calQ$ a vector field on $\calM$, $\omega$ a symplectic form on $\calM$ and $\calS$ a function on $\calM$ such that $\mathrm{gh}(\omega)=k-1$, $\mathrm{gh}(\calS)=\mathrm{gh}(\omega)+1$, $\mathrm{gh}(\calQ)=1$, and $L_{[\calQ,\calQ]}\omega=0$. 
\end{defn}
An example of a weak BV manifold is a BV manifold, where, in addition, $\calQ$ is required to be the Hamiltonian vector field of $\calS$ and $[\calQ,\calQ]=0$. We will see that a weak BV manifold contains anyway enough structure to be worth studying. The failure of $\calQ$'s being the Hamiltonian vector field of $\calS$ is what typically happens in a field theory with boundary and gives rise to the BV-BFV formalism. The failure of the condition $[\calQ,\calQ]=0$ is what occurs in the equivariant BV formalism, where the condition is replaced by the weaker one $L_{[\calQ,\calQ]}\omega=0$.

\begin{lem}
    If $k\not=-1$, we get that $[\calQ,\calQ]$ is Hamiltonian w.r.t. $\omega$.
\end{lem}

\begin{proof}
    Let $E$ be the Euler vector field and define $\calH:=\iota_E\iota_{[\calQ,\calQ]}\omega$. Then $\delta \calH=\delta\iota_E\iota_{[\calQ,\calQ]}\omega=(k+1)\iota_{[\calQ,\calQ]}\omega$.
\end{proof}

\begin{rem}
    If we consider a BF$^k$V theory \cite{CMR1}, we never consider the case of $k=-1$. Usually, we have $k\geq0$.
\end{rem}

    To motivate the above convention, suppose that $\iota_\calQ\omega=\delta\calS$ and $\frac{1}{2}(\calS,\calS)=\calT$, i.e., equation \eqref{e:SS2T}. Then $\calT=\frac{1}{2}L_\calQ\calS=\frac{1}{2}\iota_\calQ\delta \calS$. Hence, we get that \eqref{e:SS2T} is equivalent to
    \[
    \frac{1}{2}\iota_\calQ\iota_\calQ\omega=\calT.
    \]
    In turn, this is equivalent to
    \begin{equation}\label{e:QQomegadeltaT}
        \frac{1}{2}\iota_{[\calQ,\calQ]}\omega = -\delta \calT,
    \end{equation}
     a characterization of the violation of the CME as in \eqref{e:SS2T} that will be more appropriate to treat the case when a boundary is present. To prove this, simply observe that we have 
    \begin{multline*}
    \frac{1}{2}\iota_{[\calQ,\calQ]}\omega=\frac{1}{2}[L_\calQ,\iota_\calQ]\omega=\frac{1}{2}L_\calQ\iota_\calQ\omega-\underbrace{\frac{1}{2}\iota_\calQ L_\calQ\omega}_{=0}
    =\frac{1}{2}L_\calQ\delta\calS=-\frac{1}{2}\delta L_\calQ \calS, 
    \end{multline*}
    where in the first equality we have used the following
\begin{lem}\label{lem:relation_1}
    We have $\underbrace{[L_\calQ, \iota_\calQ]}_{=L_\calQ\iota_\calQ-\iota_\calQ L_\calQ}=\iota_{[\calQ,\calQ]}$.
\end{lem}

\begin{proof}
    Since we are proving an equality of derivations, it is enough to check it on functions and on exact 1-forms.
    Note that for any function $f$ we get $[L_\calQ,\iota_\calQ]f=0=\iota_{[\calQ,\calQ]}f$. Hence, we get $[L_\calQ,\iota_\calQ]=\iota_{[\calQ,\calQ]}$ on functions. Moreover, we have
    \begin{multline*}
        L_\calQ\iota_\calQ\delta f=L_\calQ L_\calQ f-\frac{1}{2}L_{[\calQ,\calQ]}f=\frac{1}{2}\iota_{[\calQ,\calQ]}\delta f=\iota_\calQ L_\calQ \delta f\\
        =-\iota_\calQ\delta L_\calQ f=-L_\calQ L_\calQ f+\underbrace{\delta\iota_\calQ L_\calQ f}_{=0}=-\frac{1}{2}L_{[\calQ,\calQ]}f=-\frac{1}{2}\iota_{[\calQ,\calQ]}\delta f.
    \end{multline*}
    Thus, we get that $[L_\calQ,\iota_\calQ]=\iota_{[\calQ,\calQ]}$ on exact 1-forms. 
\end{proof}

\subsection{Classical Equivariant Setting}
In the presence of a boundary, the equation $\iota_\calQ\omega=\delta\calS$ is usually violated, so we define $\underline{\alpha}:=\iota_\calQ\omega-\delta\calS$ and $\underline{\omega}:=\delta\underline{\alpha}$. Then we have that $\mathrm{gh}(\underline{\alpha})=k$.

In general, $\calQ$ is not cohomological. However, the fundamental assumption for what we call boundary-compatible equivariant BV formalism  is that $[\calQ,\calQ]$ is Hamiltonian, i.e., we assume that \eqref{e:QQomegadeltaT} still holds. 

\begin{prop}
    We have that $\underline{\omega}=-L_\calQ\omega$.
\end{prop}

\begin{proof}
    Note that 
    \[
    \underline{\omega}=\delta\iota_\calQ\omega=(\delta\iota_\calQ-\iota_\calQ\delta)\omega=-L_\calQ\omega.
    \]
\end{proof}

\begin{lem}\label{lem:Lemma_1}
    We have that $L_\calQ\underline{\omega}=0$.
\end{lem}

\begin{proof}
    Note that
    \[
    L_\calQ\underline{\omega}=-L_\calQ L_\calQ\omega=-\frac{1}{2}[L_\calQ,L_\calQ]\omega=-\frac{1}{2}L_{[\calQ,\calQ]}\omega=0.
    \]
\end{proof}

\begin{lem}
    If $k\not=-1$, there exists a unique $\underline{\calS}$ such that $\iota_\calQ\underline{\omega}=\delta\underline{\calS}$.
\end{lem}

\begin{proof}
    Let $E$ be the Euler vector field and define $\calH:=\iota_E\iota_\calQ\underline{\omega}$. Then, we get
    \[
    \delta \calH=L_E\iota_\calQ\underline{\omega}=(k+1)\iota_\calQ.
    \]
    Hence, we have $\underline{\calS}=\frac{1}{k+1}\calH$ with $\mathrm{gh}(\underline{\calS})=k+1$.
\end{proof}

\begin{lem}
    We have 
    \[
    \delta L_\calQ\calS=\delta\left(2\underline{\calS}-\iota_\calQ\underline{\alpha}+2\calT\right).
    \]
\end{lem}

\begin{proof}
    Note that we have
    \begin{multline*}
        \delta\underline{\calS}=\iota_\calQ \underline{\omega}=-\iota_\calQ L_\calQ\omega=-[\iota_\calQ,L_\calQ]\omega-L_\calQ\iota_\calQ\omega\\
        =\iota_{[\calQ,\calQ]}\omega-L_\calQ(\underline{\alpha}+\delta\calS)=\iota_{[\calQ,\calQ]}\omega-\iota_\calQ\delta\underline{\alpha}+\delta\iota_\calQ\underline{\alpha}+\delta L_\calQ \calS\\
        =\iota_{[\calQ,\calQ]}\omega-\iota_\calQ\underline{\omega}+\delta\iota_\calQ\underline{\alpha}+\delta L_\calQ\calS=-2\delta \calT+\delta \iota_{\calQ}\underline{\alpha}+\delta L_\calQ\calS-\delta\underline{\calS}.
    \end{multline*}
    Hence, we get that $2\delta \underline{\calS}=-2\delta \calT+\delta\iota_\calQ\underline{\alpha}+\delta L_\calQ\calS$.
\end{proof}

\begin{prop}
    If $k\not=-1$, we get that
    \[
    L_\calQ\calS=2\underline{\calS}-\iota_{\calQ}\underline{\alpha}+2\calT.
    \]
\end{prop}

\begin{proof}
This follows from the fact that $\mathrm{gh}(L_\calQ\calS)=k+1$.
\end{proof}

\begin{lem}
    We have that
    \[
    \delta\left(\frac{1}{2}\iota_\calQ\iota_\calQ\omega\right)=\delta\left(\underline{\calS}+\calT\right).
    \]
\end{lem}

\begin{proof}
    Note that we have
    \begin{multline*}
        \delta\iota_\calQ\iota_\calQ\omega=\delta\iota_\calQ\left(\underline{\alpha}+\delta\calS\right)=\delta\iota_\calQ\underline{\alpha}+\delta L_\calQ\calS\\
        =\delta\iota_\calQ\underline{\alpha}+\delta\left(2\underline{\calS}-\iota_\calQ\underline{\alpha}+2\calT\right)=2\delta \underline{\calS}+2\delta \calT
    \end{multline*}
\end{proof}

\begin{thm}
    If $k\not=-1$, we get that
    \[
    \frac{1}{2}\iota_\calQ\iota_\calQ\omega=\underline{\calS}+\calT.
    \]
\end{thm}

\begin{proof}
    This follows from the fact that $\mathrm{gh}(\underline{\calS})=k+1$.
\end{proof}

\begin{rem}
    We will use this formula as
    \[
    \frac{1}{2}\iota_\calQ\iota_\calQ\omega-\calT=\underline{\calS},
    \]
    which shows that the modification of the equivariant CME $\frac{1}{2}\iota_\calQ\iota_\calQ\omega=\calT$ is the boundary term $\underline{\calS}$.
\end{rem}

\begin{prop}
    We have that
    \[
    \frac{1}{2}L_\calQ\underline{\calS}=-L_\calQ \calT.
    \]
\end{prop}

\begin{proof}
    Note that we have
    \begin{multline*}
        L_\calQ\underline{\calS}=\iota_\calQ\delta\underline{\calS}=\iota_\calQ\iota_\calQ\underline{\omega}=-\iota_\calQ\iota_\calQ L_\calQ\omega=-\iota_\calQ [\iota_\calQ,L_\calQ]\omega-\iota_\calQ L_\calQ\iota_\calQ\omega\\
        =\iota_\calQ\iota_{[\calQ,\calQ]}\omega-[\iota_\calQ,L_\calQ]\iota_\calQ\omega-L_\calQ\iota_\calQ\iota_\calQ\omega=-2\iota_\calQ\delta \calT+\iota_{[\calQ,\calQ]}\iota_\calQ\omega-L_\calQ\left(2\underline{\calS}+2\calT\right)\\
        =-2\iota_\calQ\delta \calT-2\iota_\calQ\delta \calT-L_\calQ \left(2\underline{\calS}+2\calT\right)=-2L_\calQ \calT-2L_\calQ \calT-2L_\calQ\underline{\calS}-2L_\calQ \calT.
    \end{multline*}
    Hence, we get that
    \[
    L_\calQ\underline{\calS}=-2L_\calQ \calT.
    \]
\end{proof}

\begin{defn}
    We define $\underline{\calT}:=-L_\calQ \calT$, which implies the equation
    \begin{equation}\label{eq:boundary_BFV_eq_CME}
    \frac{1}{2}L_\calQ\underline{\calS}=\underline{\calT}.
    \end{equation}
    We call this the \emph{boundary BFV equivariant CME}.
\end{defn}

\begin{rem}
Using the boundary Poisson bracket, we can also rewrite \eqref{eq:boundary_BFV_eq_CME} as\footnote{$\underline{\omega}$ is actually degenerate, but this bracket is well-defined because $\underline{\calS}$ possesses a Hamiltonian vector field.} 
\[
\frac{1}{2}\{\underline{\calS},\underline{\calS}\}=\underline{\calT}.
\]
\end{rem}

\begin{lem}
    We have that
    \[
    \frac{1}{2}\iota_{[\calQ,\calQ]}\underline{\omega}=-\delta \underline{\calT}.
    \]
\end{lem}

\begin{proof}
    Note that we have
    \begin{multline*}
        \frac{1}{2}\iota_{[\calQ,\calQ]}\underline{\omega}=\frac{1}{2}[L_\calQ,\iota_\calQ]\underline{\omega}=\frac{1}{2}L_\calQ\iota_\calQ\underline{\omega}-\underbrace{\frac{1}{2}\iota_\calQ L_\calQ\underline{\omega}}_{=0}\\
        =\frac{1}{2}L_\calQ \delta\underline{\calS}=-\frac{1}{2}\delta L_\calQ \underline{\calS}=-\delta \underline{\calT}.
    \end{multline*}
\end{proof}

\subsection{Reduction to the Boundary} 
We now want to turn  to the construction of the symplectic form on the boundary fields. 

\begin{defn}
    We define the kernel of $\underline{\omega}$ as
    \[
    \ker\underline{\omega}:=\{Y\in\mathfrak{X}(\calF)\mid \iota_Y\underline{\omega}=0\}.
    \]
\end{defn}

\begin{lem}
    The kernel of $\underline{\omega}$ is involutive.
\end{lem}

\begin{proof}
    This follows immediately since $\delta\underline{\omega}=0$.
\end{proof}

Assume now that the leaf space of the foliation, i.e., the quotient space $\calF^\de:=\calF/\calY$, is smooth. Moreover, if we denote by $\pi\colon \calF\to \calF^\de$ the projection, then, for a uniquely determined symplectic form $\omega^\de$, we have $\underline{\omega}=\pi^*\omega^\de$.

\begin{lem}
    The vector field $\calQ$ is projectable.
\end{lem}

\begin{proof}
    Let $Y\in\ker\underline{\omega}$. Then we get
    \[
        \iota_{[\calQ,Y]}\underline{\omega}=\pm [L_\calQ,\iota_Y]\underline{\omega}=\pm \underbrace{L_\calQ\iota_Y\underline{\omega}}_{=0}\pm \underbrace{\iota_Y L_\calQ\underline{\omega}}_{=0}=0.
    \]
    Hence we get that $[\calQ,Y]\in\ker\underline{\omega}$, which proves the statement.
\end{proof}

Denote by $\calQ^\de$ the vector field on $\calF^\de$ such that 
\[
\delta \pi\calQ=\calQ^\de.
\]

\begin{lem}
    The function $\underline{\calS}$ is basic.
\end{lem}

\begin{proof}
    Let $Y\in\ker\underline{\omega}$. Then, using Lemma \ref{lem:Lemma_1}, we get that
    \[ L_Y\underline{\calS}=\iota_Y\delta\underline{\calS}=\iota_Y\iota_\calQ\underline{\omega}=\iota_\calQ\iota_Y\underline{\omega}=0.
    \]
\end{proof}

Then we can write $\underline{\calS}=\pi^*\calS^\de$ for a uniquely determined function $\calS^\de$. Thus, we get
\[
\iota_{\calQ^\de}\omega^\de=\delta\calS^\de
\]
and 
\[
\frac{1}{2}\iota_\calQ\iota_\calQ\omega-\calT=\pi^*\calS^\de.
\]

\begin{lem}
    The function $\underline{\calT}$ is basic.
\end{lem}

\begin{proof}
    Let $Y\in\ker\underline{\omega}$. Then we get that
    \[
    L_Y \underline{\calT}=\frac{1}{2}L_Y L_\calQ \underline{\calS}=\frac{1}{2}[L_Y,L_\calQ]\underline{\calS}=\pm \frac{1}{2}L_{[Y,\calQ]}\underline{\calS}=0.
    \]
\end{proof}

Let $\calT^\de$ be the corresponding function on $\calF^\de$. Then we can write
\[
\underline{\calT}=\pi^*\calT^\de.
\]
Thus, we get that
\[
\frac{1}{2}L_{\calQ^\de}\calS^\de=\calT^\de,
\]
and
\[
L_\calQ \calT=-\pi^*\calT^\de.
\]
Moreover, we get
\[
\frac{1}{2}\big\{\calS^\de,\calS^\de\big\}=\calT^\de,
\]
where $\{\enspace,\enspace\}$ again denotes the Poisson bracket on $\calF^\de$ induced by $\omega^\de$, and 
\[
\frac{1}{2}\iota_{\left[\calQ^\de,\calQ^\de\right]}\omega^\de=-\delta \calT^\de.
\]

\subsection{Summary}
Let us summarize what we have so far. On $\calF^\de$ we have $\omega^\de$ with $\mathrm{gh}(\omega^\de)=0$, $\calQ^\de$ with $\mathrm{gh}(\calQ^\de)=1$, $\calS^\de$ with $\mathrm{gh}(\calS^\de)=1$ and $\calT^\de$ with $\mathrm{gh}(\calT^\de)=1$. Moreover, the following equations hold:
\begin{align}
    \iota_{\calQ^\de}\omega^\de&=\delta\calS^\de,\\
    \frac{1}{2}L_{\calQ^\de}\calS^\de&=\calT^\de,\\
    \frac{1}{2}\big\{\calS^\de,\calS^\de\big\}&=\calT^\de,\\
    \frac{1}{2}\iota_{\left[\calQ^\de,\calQ^\de\right]}\omega^\de&=-\delta \calT^\de.
\end{align}

Furthermore, if we denote by $\pi\colon \calF\to \calF^\de$ the projection, we get the following equations:
\begin{align}
    \delta\pi \calQ&=\calQ^\de,\\
    \frac{1}{2}\iota_\calQ\iota_\calQ\omega-\calT&=\pi^*\calS^\de,\\
    L_\calQ \calT&=-\pi^*\calT^\de.
\end{align}

Finally, if $\omega^\de=\delta\alpha^\de$ with $\underline{\alpha}=\pi^*\alpha^\de$, we get the following equations:
\begin{align}
    \iota_\calQ\omega&=\delta\calS+\pi^*\alpha^\de,\\
    \frac{1}{2}L_\calQ \calS&=\calT+\pi^*\left(\calS^\de-\frac{1}{2}\iota_{\calQ^\de}\alpha^\de\right).
\end{align}

\begin{rem}
    This construction will naturally lead to a quantization, where the boundary BFV operator $\Omega$ is given (as a first approximation) by the  standard-order (a.k.a.\ Schr\"odinger) quantization of $\calS^\de$, i.e., $\Omega=\widehat{\calS^\de}$,\footnote{Here $\widehat{\enspace}$ denotes the standard-order operator quantization.} with the property
    \[
    \Omega^2=\widehat{\calT^\de}.
    \]
    Moreover, we expect the state 
    \[
    \psi=\int_\calL \exp \frac{\I}{\hbar}\calS,
    \]
    with Lagrangian submanifold $\calL\subset \calF$ such that $\calT\vert_\calL=0$, to satisfy
    \[
    \widehat{\calT^\de}\psi=0,
    \]
    where $\widehat{\calT^\de}$ satisfies the equation\footnote{This is an extra assumption, not a general property.}
    \[
    \left(\Omega-\I\hbar\Delta\right)\widehat{\calT^\de}=0.
    \]
    We will return to this below.
\end{rem}

Now assume that $\omega^\de=\delta\alpha^\de$ with $\underline{\alpha}=\pi^*\alpha^\de$. Then we get that the following two equations are satisfied:
\begin{align*}
    \iota_\calQ\omega&=\delta\calS+\pi^*\alpha^\de,\\
    \frac{1}{2}L_\calQ \calS&=\calT+\pi^*\left(\calS^\de-\frac{1}{2}\iota_{\calQ^\de}\alpha^\de\right).
\end{align*}

\subsection{Prequantization} Now we want to turn to the quantization in the presence of boundary. We refer to \cite[Section 9.2]{CMW2022} for more details.
Assume that we can choose a polarization such that $\calF^\de=T^*\calB$. Moreover, we assume a splitting 
\[
\calF=\calY\times\calB.
\]
We can then split the cohomological vector field and the symplectic form as
\begin{align*}
    \calQ&=\calQ_\calY+\calQ_\calB,\\  \omega&=\omega_{\calY\calY}+\omega_{\calY\calB}+\omega_{\calB\calB}.
\end{align*}
Now we use $\iota_\calQ\omega=\delta\calS^f+\pi^*\alpha^f$, where $\alpha^f$ is adapted to the foliation and we have\footnote{i.e., $\alpha^f$ is the canonical 1-form of $T^*\calB$}
\begin{align*}
    \calS^f&=\calS+\pi^*f,\\
    \alpha^f&=\alpha^\de-\delta f.
\end{align*}
Furthermore, we have
\begin{align*}
    \iota_{\calQ_\calY}\omega_{\calY\calY}+\iota_{\calQ_\calB}\omega_{\calY\calB}&=\delta_\calY\calS^f,\\ \iota_{\calQ_\calY}\omega_{\calY\calB}+\iota_{\calQ_\calB}\omega_{\calB\calB}&=\delta\calS^{f}+\alpha^{f}.
\end{align*}
We assume a \emph{good} splitting, i.e., one with the property that $\iota_{\calQ_\calB}\omega_{\calY\calB}=0$. Then we get that
\[
\iota_{\calQ_\calY}\omega_{\calY\calY}=\delta_\calY\calS^f.
\]
Hence, we get
\[
\left(\calS^{f},\calS^{f}\right)_\calY=\iota_{\calQ_\calY}\iota_{\calQ_\calY}\omega_{\calY\calY}=\iota_{\calQ}\iota_\calQ\omega-\iota_{\calQ_\calB}\iota_{\calQ_\calB}\omega_{\calB\calB}.
\]
This implies that 
\[
\left(\calS^{f}, \calS^{f}\right)_\calY=2\calT+\underbrace{2\calS_\text{BFV}-\iota_{\calQ_\calB}\iota_{\calQ_\calB}\omega_{\calB\calB}}_{=:2\calS^\de},
\]
which means that
\[
\frac{1}{2}\left(\calS^{f},\calS^{f}\right)_\calY=\calT+\calS^\de.
\]
We thus get
\begin{multline}
    \Delta_\calY \exp \frac{\I}{\hbar}\calS^f=-\frac{1}{2\hbar^2}\left(\calS^f,\calS^f\right)_\calY \exp \frac{\I}{\hbar}\calS^f\\=-\frac{1}{\hbar^2}\calT\exp \frac{\I}{\hbar}\calS^f-\frac{1}{\hbar^2}\calS^\de \exp \frac{\I}{\hbar}\calS^f.
\end{multline}

If we further assume that the splitting is discontinuous, we get $\omega_{\calY\calB}=0$ and $\omega_{\calB\calB}=0$. The vanishing of $\omega_{\calY\calB}$ implies that $\calS^\de=\calS_\text{BFV}$. Thus, we get
\[
\delta_\calB \calS^f=-\alpha^f=-\sum_i p_i\dd q_i.
\]
Hence, we have
\[
\frac{\delta \calS}{\delta q_i}=-p_i. 
\]
In this case, we can define
\[
\Omega:=\widehat{\calS^f}
\]
as the Schr\"odinger quantization. That is, the operator quantization of $q\in\calB$ is given by multiplication with $q$ and the operator quantization of $p\in T^*_q\calB$ is given by $\I\hbar \frac{\delta}{\delta q}$, together with the standard ordering. Then we have
\[
\Omega \exp \frac{\I}{\hbar}\calS^f=\calS^\de \exp \frac{\I}{\hbar}\calS^f.
\]
Finally, this implies the \emph{equivariant} modified quantum master equation (emQME)
\[
\left(\Omega+\hbar^2\Delta_\calY\right)\exp \frac{\I}{\hbar}\calS^f=-\calT\exp \frac{\I}{\hbar}\calS^f.
\]
Thus, the assumptions for equivariant quantization in the BV-BFV formalism are given by the equations
\begin{align}
\Omega^2&=\widehat{\calT^\partial},\\
\widehat{\calT}\exp \frac{\I}{\hbar}\calS^f&=0,\\
\calT\vert_{\calL_{\calY, \psi}}&=0.
\end{align}

\subsection{Equivariant AKSZ Theories}
We want to apply the equivariant BV-BFV formalism to AKSZ theories. Consider as the space of fields the mapping space 
\[
\calF_M=\mathrm{Map}(T[1]M,N)
\]
for a target Hamiltonian manifold $(N,\omega=\dd \alpha, \Theta)$. The AKSZ-action $\calS_M^\mathrm{AKSZ}$ is then such that 
\[
\left(\calS_M^\mathrm{AKSZ},\calS_M^\mathrm{AKSZ}\right)=0,\quad \text{and}\quad \iota_{\calQ_M}\omega_M=\delta\calS_M^\mathrm{AKSZ},
\]
if $\partial M=\varnothing$ and such that 
\[
\iota_{\calQ_M}\omega_M=\delta \calS_M^\mathrm{AKSZ}+\pi^*\alpha^\de_{\de M},
\]
if $\de M\not=\varnothing$. Note that $\calQ_M$ is the corresponding cohomological vector field of $\calS^\mathrm{AKSZ}_M$ and $\omega_M$ is the symplectic structure on $\calF_M$ constructed out of $\omega$ as we have seen before.

We first recall the equivariant extension of an AKSZ theory in the case $\partial M=\varnothing$, as introduced in \cite{BonechiCattaneoZabzine2022}. Assume $\alpha^\de_{\de M}=\sum_ix_i\dd y_i$. For $v\in\mathfrak{X}(M)$, we define
\[
\calS_{\iota}^\mathrm{AKSZ}\coloneqq\int_M \mathbf{X}_i\iota_v\mathbf{Y}_i.
\]
Thus, we define $\widehat{\calS}_M^\mathrm{AKSZ}:=\calS_M^\mathrm{AKSZ}+u\calS_\iota^\mathrm{AKSZ}$, where $u$ is a variable of degree $+2$. When $\partial M=\varnothing$, we have
\begin{align}
    \iota_{\widehat{\calQ}_M}\omega_M&=\delta\widehat{\calS}_M^\mathrm{AKSZ},\\
    \widehat{\calQ}_M&:=\calQ_M+\calQ_\iota,\\
    \left(\widehat{\calS}_M^\mathrm{AKSZ},\widehat{\calS}_M^\mathrm{AKSZ}\right)&=u\calS_L^\mathrm{AKSZ},\\
    \calS_L^\mathrm{AKSZ}&\coloneqq\int_M\mathbf{X}_iL_v\mathbf{Y}_i,\\
    \label{eq:eq_CME}
    \frac{1}{2}\iota_{\left[\widehat{\calQ}_M,\widehat{\calQ}_M\right]}\omega_M&=u\delta \calS_L^\mathrm{AKSZ},
\end{align}
where $\widehat{\calQ}_M$, $\calQ_M$ and $\calQ_\iota$ are the Hamiltonian vector fields associated to $\widehat{\calS}^\mathrm{AKSZ}_{M}$, $\calS^\mathrm{AKSZ}_M$ and $\calS^\mathrm{AKSZ}_\iota$, respectively.
Note that, using the general notation, we have $\calT=-u\calS_L^\mathrm{AKSZ}$.
\begin{rem}\label{r:generalequiv}
    For simplicity we discuss here only the case of a single vector field $v$. The general case, also discussed in \cite{BonechiCattaneoZabzine2022}, of an involutive family $\{v_1,\dots,v_r\}$ of vector fields is treated similarly. One just has to introduce variables $u_1,\dots,u_r$ of degree $+2$ and define 
    \[
    \begin{split}
        \Tilde\calS_{\iota}^\mathrm{AKSZ}&\coloneqq\sum_{\mu=1}^r u_\mu \int_M \mathbf{X}_i\iota_{v_\mu}\mathbf{Y}_i\\
        \intertext{and}
        \Tilde\calS_L^\mathrm{AKSZ}&\coloneqq\sum_{\mu=1}^r u_\mu\int_M\mathbf{X}_iL_{v_\mu}\mathbf{Y}_i.
    \end{split}
    \]
    All the formulae above immediately generalize and, in particular, we get $\calT=-\Tilde\calS_L^\mathrm{AKSZ}$. For the rest of the paper, we will focus on the case of a single vector field.
\end{rem}

When $\de M\not=\varnothing$, we have to check whether we still have an equivariant BV theory, namely, that \eqref{eq:eq_CME} still holds. We start with computing 
\begin{align*}
\delta \calS_L^\mathrm{AKSZ}&=\int_M \delta \mathbf{X}_iL_v\mathbf{Y}_i\pm \int_M\mathbf{X}_iL_v\delta \mathbf{Y}_i\\
&=\int_M \delta \mathbf{X}_i\iota_v \dd \mathbf{Y}_i+\int_M\delta \mathbf{X}_i\dd \iota_v\mathbf{Y}_i\pm \int_M \mathbf{X}_i\iota_v\dd\delta \mathbf{Y}_i\pm \int_M \mathbf{X}_i\dd \iota_v\delta \mathbf{Y}_i.
\end{align*}
If $v$ is tangent to the boundary, then we get
\begin{multline*}
    \int_M \mathbf{X}_i\iota_v\dd\delta \mathbf{Y}_i\pm \int_M \mathbf{X}_i\dd \iota_v\delta \mathbf{Y}_i\\
    =\pm \int_M\iota_v\mathbf{X}_i\dd \delta \mathbf{Y}_i\pm \int_M\dd\mathbf{Y}_i\iota_v\delta\mathbf{Y}_i\pm \int_{\de M}\mathbf{X}_i\iota_v\delta \mathbf{Y}_i\\
    =\pm \int_M\dd\iota_v \mathbf{X}_i\delta\mathbf{Y}_i\pm \int_M \iota_v \mathbf{Y}_i\delta \mathbf{Y}_i\pm \int_{\de M}\big(\iota_v\mathbf{X}_i\delta\mathbf{Y}_i+\mathbf{X}_i\iota_v\delta\mathbf{Y}_i\big).
\end{multline*}
In this case, \eqref{eq:eq_CME} holds. 

Note that in general, 
the integration by parts gives rise to the term $\int_{\de M}\mathbf{X}_i\iota_\partial^*(\iota_v\delta \mathbf{Y}_i)$, where $\iota_\partial$ denotes the inclusion map $\partial M\hookrightarrow M$.
If $v$ were not tangent to the boundary, in addition to the good term displayed above, this would also contain a term
$\int_{\de M}\mathbf{X}_i\iota_\partial^*(v^n\delta \mathbf{Y}_i^n)$, where $n$ denotes the transversal components. In conclusion, we have the
\begin{thm}
    An equivariant AKSZ theory defined in terms of a vector field $v$ is boundary-compatible if and only if $v$ is tangent to the boundary.
\end{thm}
Note that in the general case of Remark~\ref{r:generalequiv} the theory is boundary-compatible if and only if all the $v_\mu$s are tangent to the boundary.

From now on we then assume that $v$ is tangent to the boundary.
We have 
\[
\iota_{\widehat{\calQ}_M}\omega_M=\iota_{\mathcal{Q}_M}\omega_M+u\iota_{\calQ_\iota}\omega_M=\delta \calS_M^\mathrm{AKSZ}+\pi_M^*\alpha^\de_{\de M}+u\calS_\iota^\mathrm{AKSZ}=\delta\widehat{\calS}_M+\pi_M^*\alpha^\de_{\de M},
\]
where $\alpha^\de_{\de M}$ is the Noether 1-form of AKSZ theories, and $\calQ^\de_{\de M}$ is the cohomological vector field on the boundary of AKSZ theories. Thus, we have
\[
\widehat{\calS}^{\de, \mathrm{AKSZ}}_{\de M}=\calS^{\de, \mathrm{AKSZ}}_{\de M}+\calS^{\de, \mathrm{AKSZ}}_\iota.
\]
We can now compute $L_{\widehat{\calQ}_M}\calS_L^\mathrm{AKSZ}$. We get
\[
L_{\widehat{\calQ}_M}\calS_L^\mathrm{AKSZ}=L_{\calQ_M} \calS_L^\mathrm{AKSZ}+u\underbrace{L_{\calQ_\iota}\calS_L^\mathrm{AKSZ}}_{\big(\calS_\iota^\mathrm{AKSZ},\, \calS_L^\mathrm{AKSZ}\big)=0}=\pm \int_{\de M}\mathbf{X}_iL_v\mathbf{Y}_i.
\]
This implies that
\[
\frac{1}{2}\left\{\widehat{\calS}^{\de, \mathrm{AKSZ}}_{\de M},\widehat{\calS}^{\de, \mathrm{AKSZ}}_{\de M}\right\}=u\widehat{\calS}^{\de, \mathrm{AKSZ}}_L,
\]
which yields $\calT^\partial=u\widehat{\calS}^{\de, \mathrm{AKSZ}}_L$.

\subsection{Equivariant Abelian $BF$ Theory}
$BF$ theory is an AKSZ theory, so we explicitly have
\begin{align*}
\Hat\calS_M&=\int_M\big(\langle \mathbf{B},\dd \mathbf{A}\rangle
+u\langle \mathbf{B},\iota_v \mathbf{A}\rangle\big),\\
\calT_M&=-u\int_M \langle \mathbf{B},L_v \mathbf{A}\rangle,\\
\Hat\calS_{\partial M}^\partial&=\int_{\partial M}\big(\langle \mathbf{B},\dd \mathbf{A}\rangle
+u\langle \mathbf{B},\iota_v \mathbf{A}\rangle\big),\\
\calT_{\partial M}^\partial&=-u\int_{\partial M} \langle \mathbf{B},L_v \mathbf{A}\rangle.
\end{align*}
We split the boundary $\partial M$ into two disjoint components $\partial_1 M$ and $\partial_2 M$. We choose the $\mathbf{A}$-polarization on one of them and the $\mathbf{B}$-polarization on the other. By the Schr\"odinger quantization we get
\begin{align*}
\Omega &= \ii\hbar\left(
\int_{\partial_1 M} \big(\dd\mathbf{A}+u\iota_v\mathbf{A}\big)\frac\delta{\delta \mathbf{A}}
+
\int_{\partial_2 M} \big(\dd\mathbf{B}+u\iota_v\mathbf{B}\big)\frac\delta{\delta \mathbf{B}}
\right),\\
\widehat{\calT^\partial} &= u\ii\hbar\left(
\int_{\partial_1 M} L_v\mathbf{A}\frac\delta{\delta \mathbf{A}}
+
\int_{\partial_2 M} L_v\mathbf{B}\frac\delta{\delta \mathbf{B}}.
\right)
\end{align*}
We clearly have $\Omega^2=\widehat{\calT^\partial}$.

To compute the state $\psi$, we have to pick a $v$-invariant metric and then proceed as in \cite[Section 4]{CMR2} using the Hodge decomposition. In particular, we get a propagator $\eta$ which satisfies
\[
\dd\eta = \chi_i\otimes\chi^i,\qquad L_v\eta=0,
\]
where $\{\chi_i\}$ and $\{\chi^i\}$ are the chosen bases of harmonic forms relative to the boundary. We have $L_v\chi_i = L_v\chi^i=0$ for all $i$. Finally, we define
\[
\mathbf{a} = z^+_i\chi^i,\qquad \mathbf{b} = z^i\chi_i,
\]
where $\{z_i,z^+_i\}$ are the Darboux coordinates for the BV space of residual fields. With these notations, we have $\psi=\tau\exp\frac\ii\hbar \calS_\text{eff}$ with
\begin{multline*}
\calS_\text{eff} = \int_{\partial_1 M} \mathbf{A}\mathbf{b} +
\int_{\partial_2 M} \mathbf{a}\mathbf{B} +
\int \mathbf{A}\eta\mathbf{B} + u \int_M \mathbf{b}\iota_v\mathbf{a}\\
+\sum_{n>1}u^n\left(
\int \mathbf{A}\eta\iota_v\eta\cdots\iota_v\eta\mathbf{B}
+\int \mathbf{A}\eta\iota_v\eta\cdots\iota_v\eta\mathbf{b}
+\int \mathbf{a}\eta\iota_v\eta\cdots\iota_v\eta\mathbf{B}
\right)
\end{multline*}
and $\tau$ the Ray--Singer torsion of $M$ with boundary.
One can easily check that
\[
\widehat{\calT^\partial}\psi=0,\qquad \big(\Omega+\hbar^2\Delta_\text{res}\big)\psi=0,
\]
where $\Delta_\text{res}$ denotes the BV Laplacian on residual fields.

\section*{Conflict of interest and data availability}
On behalf of all authors, the corresponding author states that there is no conflict of interest and that the manuscript has no associated data.

\printbibliography

\end{document}